\documentclass[runningheads]{llncs}

\usepackage{graphicx}
\usepackage{amssymb}
\usepackage{mathtools}
\usepackage{braket}
\usepackage{xcolor}
\usepackage{soul}
\usepackage{url}
\usepackage{tensor}
\usepackage{bm}
\usepackage{array}

\newif\ifanonymous
\anonymousfalse     

\begin{document}
\title{Symmetric Cryptography via Dirichlet Convolution}

\ifanonymous
\author{Anonymous}
\authorrunning{Anonymous}
\institute{Affiliation withheld for anonymous review}
\else
\author{
Artur Sowa\inst{1} \and
Elias Hassani Sangani\inst{2} \and
Francis Bui\inst{2} \and
Grant Harris\inst{3} \and
Jonathan Norton\inst{4}
}
\authorrunning{A. Sowa et al.}

\institute{
Department of Mathematics and Statistics,
University of Saskatchewan, Saskatoon, Canada
\and
Department of Electrical and Computer Engineering,
University of Saskatchewan, Saskatoon, Canada
\and
Department of Physics and Engineering Physics,
University of Saskatchewan, Saskatoon, Canada
\and
Department of Surgery,
University of Saskatchewan, Saskatoon, Canada
}
\fi
\date{}
\maketitle
\newtheorem{algorithm}{Algorithm}[section]

\begin{abstract}

We discuss a new symmetric-key cipher for digital data encryption. Its implementation is fast, memory efficient, and resilient against classical, AI-assisted, and quantum attacks.

Let $x$, $k$, and $c$ be elements of a finite abelian group $G$ with operation $+$ and neutral element $0$. Suppose one is given ciphertext $c=x+k$. Retrieving the plaintext $x=c-k$ from the ciphertext is trivial when the key $k$ is known. Without the key, however, recovering $x$ amounts to a blind search. Even if an efficient criterion were available for identifying the plaintext among all group elements, the average number of trials would remain prohibitively large when the group is sufficiently large. The challenge in constructing a practical cryptosystem is therefore to identify a sufficiently rich group $G$, together with efficient algorithms that enable an immersion of \emph{real digital data} into $G$ and efficient evaluation of the group operations.

The proposed cipher is based on generalized Dirichlet convolution over finite fields. We investigate its security against plaintext attacks from an information-theoretic perspective. The nonlocal structure of Dirichlet convolution naturally permits plaintext salting, providing protection against plaintext attacks while allowing controlled key reuse. We establish perfect masking properties over finite fields and prove that the resulting cipher achieves information-theoretic security under key reuse when the underlying arithmetic is carried out in $\mathbb{F}_p$ with $p>2$. For arithmetic over $\mathbb{F}_2$, some information leakage remains, leading to a finite key-reuse budget.

  \vspace{.5cm}


\end{abstract}

\section{Introduction}

Encryption by means of a group operation,
\[
c=x+k,
\]
as indicated in the abstract, may derive its security from several distinct
principles. First, the key \(k\in G\) may be selected uniformly at random from
the group \(G\). Second, a blind search for the key may require an unrealistic
amount of computational time. Third, if every \(x\in G\) is a possible
plaintext, then decryption without any \emph{a priori} characterization of
\(x\) is equivalent to a blind guess.

These principles are all present in the one-time pad (OTP), where \(G\) is the
direct sum of a large number of copies of \(\mathbb{Z}_2\). The OTP provides
information-theoretic security: if the key is selected uniformly at random,
then
\[
I(x;c)=0,
\]
so that the ciphertext conveys no information about the plaintext. Moreover,
the ciphertext is itself uniformly distributed and therefore
indistinguishable from random noise. The practical limitations of the OTP are
equally well known. The key must be at least as long as the plaintext, and its
length may therefore reveal the size of the message. Furthermore, every
element of \(G\) has order \(2\), so that
\[
c+k=x+k+k=x.
\]
Consequently, reuse of the same key creates immediate vulnerabilities, and in
principle a fresh key must be supplied for every message packet.

In this work, we consider an encryption scheme based on Dirichlet convolution
over finite fields. For simplicity, a substantial portion of the algebraic
development is first carried out in characteristic \(2\), although the
principal information-theoretic results concern arithmetic in
\(\mathbb{F}_p\) for \(p>2\). There is a strong incentive to work in odd prime
characteristic: the encryption key may then be selected at random subject to
the requirement that all of its coefficients be nonzero, a clear
impossibility over \(\mathbb{F}_2\).

The proposed method retains certain features of the OTP. In characteristic
\(2\), the underlying group again consists of arrays of bits, say of length
\(N\), and the encryption operation
\[
\zeta=\kappa*\epsilon
\]
is abelian; see~\eqref{DCMpEncrypt}. The crucial difference is that Dirichlet
convolution acts nonlocally. A ciphertext coordinate generally depends on
several coordinates of both the plaintext and the key, according to the
divisor structure of its index. This nonlocality permits the plaintext to be
concatenated with randomly generated components, which we call
\emph{salting}. Randomness introduced in the leading coordinates can then
propagate through the convolution into later ciphertext coordinates.

This mechanism leads to the first principal theme of the paper:
information-theoretic masking. Over \(\mathbb{F}_p\), with \(p>2\), and under
the natural assumption that the key coefficients are nonzero, sufficiently
salted ciphertext coordinates are uniformly distributed and statistically
independent of the plaintext. In particular, their mutual information with
the plaintext vanishes. The principal result is stated in
Theorem~\ref{Affine_Uniformity}. Remarkably, this masking persists under key
reuse, provided that fresh salt is supplied for each encryption. The binary
case is exceptional: over \(\mathbb{F}_2\), the masking terms need not be
perfectly uniform, and repeated use of the same key may produce a gradual
accumulation of information.

The linearity of the encryption transformation in the plaintext also raises
an immediate cryptanalytic question. One might naturally suspect that an
adversary who can substitute plaintexts and observe the corresponding
ciphertexts could use carefully chosen probe vectors, analogous to Dirac
delta inputs, to reveal the underlying convolution operator and ultimately
the secret key. The second principal theme of the paper is a detailed
analysis of this apparent vulnerability.

Plaintext substitution does indeed expose a rich affine structure in the
ciphertext space and gives rise to a repeated-input key-recovery attack.
Nevertheless, the attack encounters intrinsic arithmetic blind spots. We
identify a sharp phase transition at the threshold
\[
N=2k,
\]
where \(k\) is the length of the random salt. In the regime \(N\ge 2k\), all
leading key coefficients are recoverable from the repeated-input equations.
When \(N<2k\), certain coefficients become invisible because none of their
multiples occurs in the relevant observation window. Subsection~\ref{section_invisibility} develops an arithmetic characterization of the visible and invisible indices in terms of divisibility and proves that the invisible coefficients are precisely the free variables of the recovery system. With suitable parameter choices, the resulting ambiguity remains sufficiently large after any practically realizable amount of key reuse to render complete key recovery by this class of attacks impractical.

Section~\ref{Section_Cycles} develops the algebraic framework used throughout the paper. While the group-theoretic arguments are elementary, they provide a self-contained treatment of the finite Dirichlet group, establish the notation adopted throughout the paper, and collect structural results for which we have not found a convenient reference. In particular, we determine the decomposition of the finite Dirichlet group into cyclic components, culminating in Propositions~\ref{Prop_cycles} and~\ref{Prop_cycles_p}, corresponding to the binary and general prime-characteristic cases, respectively. In contrast with the binary OTP group, whose nonzero elements all have order \(2\), the orders of Dirichlet group elements may be substantially larger and grow with the truncation parameter \(N\).

Following the algebraic preliminaries in
Section~\ref{Section_Cycles}, Section~\ref{Section_DCM}
introduces encryption by Dirichlet convolution together with the
associated salting mechanism. The subsequent sections develop the
information-theoretic and cryptanalytic aspects of the construction,
while the final sections specialize the theory to digital image
encryption over the finite field \(\mathbb{F}_{257}\).

The present approach is fundamentally information-theoretic rather than computational. Following Shannon's formulation of perfect secrecy \cite{Shannon1949}, the objective is to eliminate exploitable statistical dependence between the plaintext and the ciphertext, rather than to rely solely on the computational difficulty of recovering such dependence. Consequently, whenever the ciphertext distribution is proved to be independent of the plaintext, no subsequent classical, quantum, or AI-assisted algorithm can recover information that is not statistically present. The computational aspects of the proposed scheme are therefore complemented by rigorous information-theoretic guarantees.

\section{The finite Dirichlet algebra over $\mathbb{F}_p$}\label{Section_Cycles}

For clarity, we begin with the special case \(p=2\), since the extension to an arbitrary prime field is straightforward. The general result is stated in Subsection~\ref{subection_cycle_p}. 

Consider \[
\mathcal{L}_{N} = \{ \bm{\epsilon} = [\epsilon_1, \epsilon_2, \epsilon_3, \ldots, \epsilon_N]:\quad  \epsilon_j \in \mathbb{F}_2 \,\,
\mbox{ for all } j \},
\] 
i.e. the set of all $N$ bit binary sequences $[\epsilon_1, \epsilon_2, \epsilon_3, \ldots, \epsilon_N]$. Clearly, $\mathcal{L}_{N}$   is a linear space over the field $\mathbb{F}_2$ when endowed with the usual bit-wise addition modulo 2 (XOR) and multiplication by a scalar (zeros or one). In addition, we introduce an operation of multiplication in $\mathcal{L}_{N}$ via Dirichlet convolution modulo 2 (DCM2). Namely, for any $\bm{\alpha}, \bm{\beta} \in \mathcal{L}_{N}$ we define $\bm{\gamma} = \bm{\alpha}*\bm{\beta}$ as follows:
\begin{equation}\label{mult_eps}
  \gamma_n = \sum_{d|n} \alpha_{d} \beta_{n/d}\,\, \pmod{2}, \quad n = 1, 2,\ldots N.
\end{equation}
Clearly, the multiplication $*$ is commutative. 
This makes $\mathcal{L}_{N}$  into an algebra over $\mathbb{F}_2$. The group of this algebra's invertible elements (with respect to $*$) is 
 \[
\mathcal{L}^*_{N} = \{ \bm{\epsilon} = [1, \epsilon_2, \epsilon_3, \ldots, \epsilon_N]:\quad  \epsilon_j \in \mathbb{F}_2 \,\,
\mbox{ for all } j \geq 2 \}.
\] 
Indeed,  for $\bm{\alpha} \in \mathcal{L}^*_{N}$ its inverse $\bm{\beta} = \bm{\alpha}^{*-1}$ has $\beta_1 = 1$ and remaining bits determined via the recurrence: 
\begin{equation}\label{mult_eps_inv}
  \beta_n = \sum_{d|n: d<n} \beta_{d} \alpha_{n/d}\,\, \pmod{2}\quad \mbox{ for all } n>1.
\end{equation}
Note that $\dim \mathcal{L}_{N} = 2^N$, so that in the sequence of algebras 
\begin{equation}\label{alebra_seq}
  \mathcal{L}_{2} \subset \mathcal{L}_{3} \subset \mathcal{L}_{4} \subset \ldots,
\end{equation}
the dimension doubles at each step.  
Also,
\begin{equation}\label{group_seq}
  \mathcal{L}^*_{2} \subset \mathcal{L}^*_{3} \subset \mathcal{L}^*_{4} \subset \ldots,
\end{equation}
and the group order doubles at every step, as for $\mathcal{L}^*_{N}$ it is $2^{N-1}$.

\subsection{The structure of $\mathcal{L}^*_{N}$}

 It is a general fact that every finite abelian group can be expressed as a direct sum of cyclic groups of prime power order. We will be more specific for $\mathcal{L}^*_{N}$. To this end, we utilize the canonical basis of the ambient algebra. Namely, let
 \[
\bm{e}_k= 
\begin{bmatrix}
0 & 0 & \cdots & 0 & \underset{k^{\text{th}}}{1} & 0 & \cdots & 0
\end{bmatrix}.
\]
This enables us to write all elements of the group, in the form:
\begin{equation}\label{delta_basis}
  \bm{\epsilon} = \bm{e}_1  + \sum_{k=2}^{N} \epsilon_k \bm{e}_k .
\end{equation}
It follows directly from (\ref{mult_eps}) that
\begin{equation}\label{mult_b}
  \bm{e}_k*\bm{e}_l = \left\{
\begin{array}{ll}
\bm{e}_{kl}  & \text{if } kl\leq N \\
0 & \text{if } kl > N
\end{array}
\right.
\end{equation}
The elements
\[
 \bm{e}_1 + \bm{e}_p \quad \mbox{ where } p \leq N \mbox{ is prime}
\]
generate a subgroup of $\mathcal{L}^*_{N}$. We will refer to these elements as prime generators. 
The order of $\bm{e}_1 + \bm{e}_p$, denoted $m_p(N)$, is the smallest positive integer $m$, such that  
\begin{equation}\label{delta2_order}
  (\bm{e}_1 + \bm{e}_p) ^{*m} = \bm{e}_1 \quad \mbox{ in } \mathcal{L}^*_{N} . 
\end{equation}
Note that (\ref{mult_b}) implies a binomial formula:
\[
(\bm{e}_1 + \bm{e}_p) ^{*m} = \bm{e}_1 + \sum_{k=1}^{m-1} \binom{m}{k} \bm{e}_{p ^k}+ \bm{e}_{p ^{m}} \quad \pmod{2}.
\]
Therefore, for (\ref{delta2_order}) to hold, one must have $p^m >N$, i.e. $m > log_p(N)$. Also, all the coefficients $\binom{m}{k}$ listed above must be even, which implies that $m$ must be a power of $2$. Thus, $m_p(N)$ is the smallest power of $2$ exceeding $\log_pN$, i.e.  
\begin{equation}\label{orderm}
  m_p(N) = 2^{\lfloor \log_2\log_p(N) \rfloor +1 }.
\end{equation}
Note that the order of elements $\bm{e}_1 + \bm{e}_p$ decreases with increasing $p$. Here is a complete table for $N=256$: 
\[
\begin{array}{|c|l|}
\hline
m_p & \text{Primes } p < 256  \\
\hline
16 & 2 \\
8  & 3 \\
4  & 5\text{--}13 \\
2  & 17\text{--}251 \\
\hline
\end{array}
\]
Utilizing this list of generators, we obtain subgroups:
  \begin{equation}\label{DN_alg_structure}
 \bigoplus\limits_{p: p\leq N} \mathbb{Z}_{m_p(N)} \lhd  \mathcal{L}^*_{N} .
\end{equation} 
Note that, depending on $N$, the elements $\bm{e}_1 + \bm{e}_p$ need not be a complete set of generators. It is instructive to consider a few examples: 

\begin{table}[ht!]
\centering
\begin{tabular}{|l|l|l|}
\hline
\textbf{\( N \)} & \textbf{ \( \mathcal{L}_N^* \)} & \textbf{Generators (p), \textbf{(n)}} \\
\hline
2 & \( \mathbb{Z}_2 \) & 2 \\
3 & \( \mathbb{Z}_2 \oplus \mathbb{Z}_2 \) & 2, 3 \\
4 & \( \mathbb{Z}_4 \oplus \mathbb{Z}_2 \) & 2, 3 \\
5 & \( \mathbb{Z}_4 \oplus \mathbb{Z}_2 \oplus \mathbb{Z}_2 \) & 2, 3, 5 \\
6 & \( \mathbb{Z}_4 \oplus \mathbb{Z}_2^{\oplus 3} \) & 2, 3, 5, \textbf{6} \\
7 & \( \mathbb{Z}_4 \oplus \mathbb{Z}_2^{\oplus 4} \) & 2, 3, 5, \textbf{6}, 7 \\
8 & \( \mathbb{Z}_4 \oplus \mathbb{Z}_2^{\oplus 5} \) & 2, 3, 5, \textbf{6}, 7 , \textbf{8} \\
\hline
\end{tabular}
\caption{Cyclic decomposition and generators of \( \mathcal{L}_N^* \). Non-prime generators are in bold.}
\label{Table1}
\end{table}

\begin{itemize}
\item
 $ \mathcal{L}^*_{2} = \mathbb{Z}_2$ is generated by the prime generator, $\bm{e}_1 + \bm{e}_2$.
  
  \item
  $ \mathcal{L}^*_{3} = \mathbb{Z}_2 \oplus \mathbb{Z}_2 $ is generated by prime generators, $\bm{e}_1 + \bm{e}_p$,  with $ p = 2, 3$.
    \item
  $ \mathcal{L}^*_{4} = \mathbb{Z}_4 \oplus \mathbb{Z}_2 $ is generated by prime generators, also with $ p = 2, 3$. Note that, even though, cases $N=3$ and $N=4$ have the same generators, the order of one of them increases from $2$ to $4$. 
  \item
  $ \mathcal{L}^*_{5}= \mathbb{Z}_4\oplus\mathbb{Z}_2 \oplus \mathbb{Z}_2 $ is generated by prime generators corresponding to $ p = 2, 3, 5$. In this case, we have a new prime element. 
  \item
  $\mathcal{L}^*_{6}$ is the first case, where the subgroup generated by prime generators is proper, i.e. $\mathcal{L}^*_{6} \rhd \mathbb{Z}_4\oplus\mathbb{Z}_2 \oplus \mathbb{Z}_2 $. Indeed, the list of primes does not change as we increase $N$ from $5$ to $6$. Also, the order of these generators does not increase. However, note that in this case $\bm{e}_1 + \bm{e}_6$ cannot be obtained as a product of prime generators. Indeed, a direct calculation shows that 
 \[
 (\bm{e}_1 + \bm{e}_2)^{*a}*(\bm{e}_1 + \bm{e}_3)  \neq \bm{e}_1 + \bm{e}_6\quad \mbox{ for } a = 1,2,3.
 \]     
 Therefore, to complete the list of generators for this group, we need to supplement the list of generators by $\bm{e}_1 + \bm{e}_6$, which is also an element of order $2$.  
 \item
 For $\mathcal{L}^*_{7}$, we take all the generators of $\mathcal{L}^*_{6}$, and we add one more generator: $\bm{e}_1 + \bm{e}_7$. 
 \item
   For $\mathcal{L}^*_{8}$ we add one more generator: $\bm{e}_1 + \bm{e}_8$. 
   
   These examples are summarized in Table \ref{Table1}.
   \end{itemize}
      \vspace{.2cm}
   
   The order of generators that are not prime is given by a formula analogous to (\ref{orderm}), namely:
   \begin{equation}\label{ordermn}
  m_n(N) = 2^{\lfloor \log_2\log_n(N) \rfloor +1 }.
\end{equation}
While the set of generators of a group is not unique, we construct specific sets of generators $\mbox{Gen}(N)$  as follows:  First, $\mbox{Gen}(2) = \{ \bm{e}_1 + \bm{e}_2\}$. The consecutive $\mbox{Gen}(N)$ are defined inductively, accounting for two mutually exclusive scenarios, as follows:
   \begin{enumerate}
     \item 
     
     If the order of one of the elements of $\mbox{Gen}(N)$ doubles in $\mathcal{L}^*_{N+1}$ relative to $\mathcal{L}^*_{N}$, we set $\mbox{Gen} (N+1) = \mbox{Gen}(N)$. In other words, the set of generators does not change. 
     
     \item 
     In the alternative case, when none of the elements of $\mbox{Gen}(N)$ double their order, we set $\mbox{Gen} (N+1) = \mbox{Gen}(N) \cup \{ \bm{e}_1 + \bm{e}_{N+1}\}$. Note that the element $\bm{e}_1 + \bm{e}_{N+1}$ has order $2$ in $\mathcal{L}^*_{N+1}$.  
    
   \end{enumerate}
Note that the first case scenario occurs when, for a certain generator $\bm{e}_1 + \bm{e}_{n} \in \mbox{Gen}(N)$, one has 
\[
m_n(N+1) = 2 m_n(N), \mbox{ equivalently }\quad \lfloor \log_2\log_n(N+1) \rfloor = \lfloor \log_2\log_n(N) \rfloor +1.
\]
This occurs precisely when $\log_2\log_n(N+1)$ is an integer, i.e. 
$N+1 = n^{2^a}$ for a certain integer $a \geq 1$, i.e. 
$N+1$ is a square. In summary, we have demonstrated the following: 
\begin{proposition}[Cyclic Decomposition]\label{Prop_cycles}
The group $\mathcal{L}^*_{N}$ is generated by the set of its elements: 
\[
\mbox{Gen}(N) = \{\bm{e}_1 + \bm{e}_{n}: 2\leq n\leq N \mbox{ and } n \mbox{ is not a square }\}
\]
The order of $\bm{e}_1 + \bm{e}_{n}$ is given by (\ref{ordermn}). Finally,
 \[
 \mathcal{L}^*_{N} = \bigoplus\limits_{n\in [2,N]: \sqrt{n}\notin \mathbb{N} } \mathbb{Z}_{m_n(N)},  
\]
(i.e., the direct sum is over all those $n\in [2,N]$ which are not a square.) 
\end{proposition}

\subsection{Generalization to Prime Characteristic}
\label{subection_cycle_p}

The preceding arguments extend with only minor modifications to
Dirichlet convolution modulo a prime \(p\). In this setting,
\(\bm{\gamma}=\bm{\alpha}*\bm{\beta}\) is defined by
\begin{equation}\label{mult_eps_gen_p}
  \gamma_n=\sum_{d\mid n}\alpha_d\beta_{n/d}\pmod p,
  \qquad n=1,2,\ldots,N,
\end{equation}
and the group of invertible elements is
\[
\mathcal{L}^{(p)*}_{N}
=
\left\{
\bm{\epsilon}
=
[\epsilon_1,\epsilon_2,\ldots,\epsilon_N]:
\epsilon_j\in\mathbb{F}_p
\text{ for all }j,
\;
\epsilon_1\neq0
\right\}.
\]

Repeating the proof of Proposition~\ref{Prop_cycles} in this setting
yields the following generalization.

\begin{proposition}[Generalized Cyclic Decomposition]\label{Prop_cycles_p}
The group $\mathcal{L}^{(p)*}_{N}$ under Dirichlet convolution modulo $p$ is generated by the set:
\[
\text{Gen}_p(N) = \left\{ \bm{e}_1 + \bm{e}_n : 2 \leq n \leq N, \text{ and } n \mbox{ is not a } p\mbox{-th power}  \right\}.
\]
Each generator \( \bm{e}_1 + \bm{e}_n \) has order:
\[
m_n^{(p)}(N) = p^{\lfloor \log_p \log_n(N) \rfloor + 1}.
\]
Therefore, the group decomposes as:
\[
\mathcal{L}^{(p)*}_{N} \cong \bigoplus_{n \in [2,N], \sqrt[p]{n}\notin \mathbb{N} } \mathbb{Z}_{m_n^{(p)}(N)}.
\]
\end{proposition}

\begin{example}
For digital image encryption, it is natural to take \(N=535\) and
\(p=257\). In this case, the group of invertible elements is
isomorphic to the direct sum of \(512\) copies of the cyclic group
\(\mathbb{Z}_{257}\), namely,
\[
\mathcal{L}^{(257)*}_{535}
\cong
\left(\mathbb{Z}_{257}\right)^{\oplus512}.
\]
Thus, every nonidentity element of this group has order \(257\).
\end{example}

\section{Dirichlet Convolution as an Encryption Primitive}
\label{Section_DCM}

Having developed the underlying algebraic framework, we now introduce
Dirichlet convolution as an encryption primitive. We describe the basic
encryption and decryption operations, explain the role of salting, and
outline the cryptographic principles that motivate the subsequent
information-theoretic and cryptanalytic analyses.

Throughout this and the following sections, the construction is formulated
over the finite field $\mathbb{F}_p$, where $p$ is an arbitrary prime,
unless the binary case $p=2$ is explicitly specified. Accordingly,
plaintexts, ciphertexts, and keys are regarded as vectors over
$\mathbb{F}_p$, with the binary setting treated as a special case.
To simplify notation, vectors will henceforth be denoted by ordinary Greek letters rather than boldface symbols.

An element
\(
\kappa\in\mathcal{L}^{(p)*}_N
\)
is selected uniformly at random and serves as the secret key.
Throughout, we assume that all random variables are drawn independently
from the uniform distribution.
Plaintexts are represented by vectors
\(
\epsilon\in\mathcal{L}^{(p)}_N.
\)
Encryption is defined by
\begin{equation}\label{DCMpEncrypt}
\epsilon\longmapsto
\kappa *\epsilon
=:
\zeta,
\qquad\mbox{(Encryption)}
\end{equation}
producing the ciphertext $\zeta$.
Decryption is performed by applying the Dirichlet inverse of the key,
\begin{equation}\label{DCMpDecrypt}
\zeta\longmapsto
\kappa^{*-1}*\zeta
=
\epsilon,
\qquad\mbox{(Decryption).}
\end{equation}

For many of the subsequent arguments, it is convenient to represent Dirichlet convolution by the associated lower triangular matrix. Accordingly, we write
\[
\zeta
=
D_{\kappa}\epsilon,
\]
where \(D_{\kappa}\) is the lower triangular matrix defined by
\begin{equation}\label{D_matrix}
  (D_{\kappa})_{ij}
=
\begin{cases}
\kappa_{i/j},
&
j\mid i,
\\
0,
&
j\nmid i.
\end{cases}
\end{equation}
Since \(\kappa_1\neq0\), the matrix \(D_{\kappa}\) is invertible, and
decryption is equivalently expressed as
\[
\epsilon
=
D_{\kappa}^{-1}\zeta.
\]

This construction resembles the one-time pad (OTP), but differs in one
fundamental respect. Whereas the OTP encrypts each coordinate
independently, Dirichlet convolution is inherently nonlocal: each
ciphertext coordinate typically depends on several coordinates of both
the plaintext and the key. This nonlocality naturally enables
\emph{salting}, whereby the message is preceded by randomly generated
coordinates. Specifically, we consider plaintexts of the form
\[
[\mu_1,\mu_2,\ldots,\mu_k,
\epsilon_1,\epsilon_2,\ldots,\epsilon_{N-k}],
\]
where the $\mu_j$ are independent uniformly distributed random elements
of $\mathbb{F}_p$, while the $\epsilon_j$ carry the information.
During encryption, the random salt is mixed with the plaintext through
Dirichlet convolution. After decryption, the salt is simply discarded.

Salting produces a different ciphertext each time, even when the same
secret key is reused with fresh random salt. Consequently, the key need
not match the length of the plaintext. For example, an image may be
encrypted by applying Dirichlet convolution successively to its rows and
columns while reusing the same key together with independently generated
salt for each encryption.

The linearity of Dirichlet convolution also suggests a natural
cryptanalytic attack. Indeed, if an adversary knew both the complete
plaintext $\epsilon$ and the corresponding ciphertext
$\bm{\zeta}$, then the key could be recovered by
\[
\kappa
=
\epsilon^{*-1}*\zeta,
\]
provided that $\epsilon$ is invertible. Salting removes precisely
the information required to perform this inversion, since the attacker
has no access to the leading random coordinates. At first sight, one
might nevertheless attempt an exhaustive search over the missing random
coordinates, assuming that the underlying plaintext could subsequently
be recognized. The extent to which repeated-input attacks of this kind
are frustrated by the resulting uncertainty is the subject of the
following sections.

We begin by analyzing the information-theoretic masking provided by
salting.

\section{Uniformity and Masking}

The security of the proposed encryption scheme derives from two distinct
sources of randomness. The first is the initial random selection of the
secret key, while the second is the fresh random salt introduced before
each encryption. Although both mechanisms induce uniformity properties,
they correspond to different probabilistic models and play complementary
roles in the cryptographic design.

The first model considers a uniformly distributed secret key and fixed
plaintext. In this setting, the encryption map acts as a bijection and
therefore preserves the uniform distribution. The second model, which
forms the foundation of the proposed encryption architecture, keeps the
secret key fixed while introducing fresh independent salt variables for
each encryption. The principal objective of this section is to show that
the randomness carried by the salt propagates through Dirichlet
convolution to produce strong masking of the plaintext.

\subsection{Uniformity Induced by a Fresh Random Key}

The following result describes the statistical behaviour of the
ciphertext when the secret key is selected uniformly at random. It
provides a useful reference point for the analysis of the fixed-key
setting considered subsequently.

\begin{theorem}[Global Uniformity under a Uniform Key]
\label{Global_Key_Uniformity}

Let
$
\alpha\in\mathbb F_p^N
$
satisfy
$
\alpha_1\neq0.
$
If the key
$
\kappa
$
is uniformly distributed in $\mathbb F_p^N$, then
\[
\zeta=\kappa*\alpha
\]
is uniformly distributed in $\mathbb F_p^N$.

If, instead, $\kappa$ is uniformly distributed in
$
\mathcal{L}^*_{N}
=
\{\kappa\in\mathbb F_p^N:\kappa_1\neq0\},
$
then $\zeta$ is uniformly distributed in $\mathcal{L}^*_{N}$.

\end{theorem}

\begin{proof}
Since Dirichlet convolution is commutative,
$
\zeta
=
\kappa*\alpha
=
D_\alpha\kappa.
$
The matrix $D_\alpha$ is lower triangular with diagonal entries
$\alpha_1$, and is therefore invertible since $\alpha_1\neq0$.
Therefore,
$
\kappa
\longmapsto
\kappa*\alpha
$
is a bijection of $\mathbb F_p^N$. Moreover, because invertibility is
preserved under Dirichlet convolution, this bijection restricts to
$\mathcal{L}^*_{N}$.
Since bijections preserve the uniform distribution on finite sets, the
assertion follows.
\end{proof}

\subsection{Uniformity Induced by Fresh Salts}

We now turn to the probabilistic model that underlies the proposed
encryption scheme. Throughout the remainder of this section, the secret
key
\(
\kappa
\)
is regarded as fixed, while each encryption employs fresh independent
salt variables.

Let the augmented plaintext be written in the form
\[
\alpha
=
\begin{pmatrix}
\mu\\
\epsilon
\end{pmatrix},
\]
where
\(
\mu\in\mathbb F_p^k
\)
contains the random salt, and
\(
\epsilon\in\mathbb F_p^{\,N-k}
\)
contains the information-bearing coordinates.

Partition the columns of the encryption matrix \(D_\kappa\) according to
the decomposition of the input vector into salt and plaintext
coordinates:
\begin{equation}
\label{Block_Decomposition}
D_\kappa
=
\begin{bmatrix}
B & R
\end{bmatrix},
\qquad
B\in M_{N\times k}(\mathbb F_p),\;
R\in M_{N\times(N-k)}(\mathbb F_p).
\end{equation}
The matrix $B$ has full column rank. Indeed, the first $k$ rows of $B$
form a lower triangular $k\times k$ submatrix. Since the diagonal
entries are all equal to $\kappa_1$, its determinant is
$
\kappa_1^k.
$
As $\kappa_1\neq0$ by the invertibility assumption on the key,
it follows that
\begin{equation}\label{rank_B}
\operatorname{rank}(B)=k.
\end{equation}
Consequently, the fresh salt contributes $k= \dim\operatorname{im}(B)$ independent dimensions of randomness to
the ciphertext.
The ciphertext therefore admits the decomposition
\begin{equation}
\label{Affine_Decomposition}
\zeta
=
B\mu
+
R\epsilon,
\end{equation}
which separates the two sources of the ciphertext. The matrix $R$ describes the deterministic contribution of
the information-bearing coordinates, while the matrix $B$ governs the
propagation of the fresh randomness introduced by the salt. The
information-theoretic properties of the encryption scheme are therefore
determined by the interaction of these two components.

The following theorem provides the fundamental probabilistic description
of the ciphertext.
\begin{theorem}[Affine Uniformity of the Ciphertext]
\label{Affine_Uniformity}
Assume that the salt vector
$
\mu
$
is uniformly distributed on $\mathbb F_p^k$ and statistically
independent of the plaintext vector
$\epsilon$.
Then, for every plaintext vector
$
\epsilon=e,
$
the conditional distribution of the ciphertext is supported on the affine
subspace
\[
Re+\operatorname{im}(B),
\]
and is uniform there. More precisely,
\[
\Pr(\zeta=z\mid\epsilon=e)
=
\begin{cases}
p^{-k},
&
z\in Re+\operatorname{im}(B),
\\[2mm]
0,
&
z\notin Re+\operatorname{im}(B).
\end{cases}
\]

\end{theorem}

\begin{proof}
By the affine decomposition \eqref{Affine_Decomposition}, conditioning
on $\epsilon=e$ leaves $B\mu$ as the random component of the ciphertext,
while $Re$ is fixed. Since $\operatorname{rank}(B)=k$, the map
\(
B:\mathbb F_p^k\to\operatorname{im}(B)
\)
is bijective. Thus $B\mu$ is uniform on $\operatorname{im}(B)$.
Translation by $Re$ then shows that $\zeta$ is uniform on
\(
Re+\operatorname{im}(B).
\)
This affine subspace contains $p^k$ points, each having conditional
probability $p^{-k}$.
\end{proof}

We have the following:

\begin{corollary}
\label{Affine_Entropy}

For every plaintext
$e\in\mathbb F_p^{N-k}$,
the conditional entropy of the ciphertext satisfies
\[
H(\zeta\mid\epsilon=e)
=
k\log p.
\]

\end{corollary}

\begin{proof}
By Theorem~\ref{Affine_Uniformity}, conditioned on
$\epsilon=e$, the ciphertext is uniformly distributed on an affine
subspace of cardinality $p^k$. Hence
\[
H(\zeta\mid\epsilon=e)
=
\log(p^k)
=
k\log p.
\]
\end{proof}

The corollary shows that every independent salt symbol contributes
exactly one symbol of conditional randomness to the ciphertext.
Consequently, the encryption process preserves all randomness introduced
by the salt: conditioned on the plaintext, the ciphertext retains
precisely the $k$ degrees of freedom contributed by the salt.

\begin{remark}[Propagation of Randomness]
\label{Remark_Propagation}
Theorem~\ref{Affine_Uniformity} identifies the matrix $B$ as the
mechanism through which the fresh salt propagates randomness throughout
the ciphertext. By~\eqref{rank_B}, the matrix $B$ has full column rank,
so the $k$ independent degrees of randomness introduced by the salt
remain linearly independent after Dirichlet convolution.

For every coordinate projection
\[
P_S:\mathbb F_p^N\rightarrow\mathbb F_p^{|S|},
\qquad
S\subset\{1,\ldots,N\},
\]
the affine subspace
\[
Re+\operatorname{im}(B)
\]
is mapped onto
\[
P_S(Re)+\operatorname{im}(P_SB).
\]
Consequently, every projected ciphertext subvector is uniformly
distributed on its projected affine subspace. In particular, if
\[
\operatorname{rank}(P_SB)=|S|,
\]
then
\[
\operatorname{im}(P_SB)=\mathbb F_p^{|S|},
\]
so the projected ciphertext is uniformly distributed on
$\mathbb F_p^{|S|}$ and is therefore independent of the plaintext.
\end{remark}

The preceding results describe the masking properties of a single
ciphertext produced under a fixed secret key. The next question is
whether these properties persist when several independently resalted
ciphertexts are observed simultaneously and combined by an adversary.
The following subsection shows that affine masking is stable under
arbitrary linear combinations of such ciphertexts, demonstrating that
the propagation of randomness is an intrinsic structural property of the
encryption scheme rather than a feature of individual encryptions.

\subsection{Stability Under Linear Combinations}
\label{Linear_Combination_Stability}

Theorem~\ref{Affine_Uniformity} describes the distribution of a single
ciphertext. We now show that the same affine masking property is
preserved under arbitrary nontrivial linear combinations of ciphertexts
encrypted with the same secret key and independently generated salts.

\begin{corollary}[Stability Under Linear Combinations]
\label{Linear_Combination_Corollary}

Let
\(
\zeta^{(j)}
=
R\epsilon^{(j)}
+
B\mu^{(j)},
\;
j=1,\ldots,m,
\)
be ciphertexts produced according to
\eqref{Affine_Decomposition}, where the salt vectors
\(
\mu^{(1)},\ldots,\mu^{(m)}
\)
are independent and uniformly distributed on
\(
\mathbb F_p^k.
\)

Let
\(
a_1,\ldots,a_m\in\mathbb F_p
\)
be coefficients that are not all zero, and define
\[
\widetilde{\zeta}
=
\sum_{j=1}^{m}a_j\zeta^{(j)},
\qquad
\widetilde{\epsilon}
=
\sum_{j=1}^{m}a_j\epsilon^{(j)}.
\]
Then
\[
\widetilde{\zeta}
=
R\widetilde{\epsilon}
+
B\widetilde{\mu},
\]
where
\(
\widetilde{\mu}
=
\sum_{j=1}^{m}a_j\mu^{(j)}
\)
is uniformly distributed on
\(
\mathbb F_p^k.
\)
Consequently,
\(
\widetilde{\zeta}
\)
is uniformly distributed on the affine subspace
\[
R\widetilde{\epsilon}
+
\operatorname{im}(B).
\]

\end{corollary}

\begin{proof}

By linearity,
\[
\widetilde{\zeta}
=
R\widetilde{\epsilon}
+
B\widetilde{\mu},
\qquad
\text{where}\qquad
\widetilde{\mu}
=
\sum_{j=1}^{m}a_j\mu^{(j)}.
\]
Choose
\(r\)
with
\(a_r\neq0\).
Conditioning on all salt vectors except
\(
\mu^{(r)},
\)
we obtain
\(
\widetilde{\mu}
=
a_r\mu^{(r)}+c,
\)
where
\(c\)
is fixed.
Since multiplication by the nonzero scalar 
\(a_r\)
and translation by
\(c\)
are bijections of
\(
\mathbb F_p^k,
\)
the vector
\(
\widetilde{\mu}
\)
is uniformly distributed on
\(
\mathbb F_p^k.
\)
The final assertion therefore follows immediately from
Theorem~\ref{Affine_Uniformity}.
\end{proof}

\section{Plaintext-Substitution Cryptanalysis}
\label{section_attack}

The preceding sections established the information-theoretic masking
properties of Dirichlet convolution encryption. We now consider an active
attack model in which the adversary substitutes plaintexts and observes
the corresponding ciphertexts. Since the encryption transformation is
linear in the plaintext, one might naturally suspect that carefully
chosen probe vectors could reveal the underlying convolution operator and
ultimately the secret key.

The following analysis shows that plaintext substitution indeed exposes a
rich affine structure leading to a repeated-input key-recovery attack.
However, the attack possesses intrinsic arithmetic blind spots that
fundamentally limit its effectiveness. As we shall see, the resulting
ambiguity is sufficiently large to render this class of attacks
impractical, thereby providing practical security against plaintext
substitution attacks.

\subsection{Affine Structure of the Ciphertext Space}
\label{Affine_Structure}

Theorem~\ref{Affine_Uniformity} shows that, for every fixed plaintext,
the ciphertext is uniformly distributed on an affine subspace of
$\mathbb F_p^N$. The global geometry of these affine subspaces follows
from the block decomposition~\eqref{Block_Decomposition}.

Since $D_\kappa$ is invertible, its columns form a basis of
$\mathbb F_p^N$. Hence
\[
\mathbb F_p^N
=
\operatorname{im}(B)
\oplus
\operatorname{im}(R).
\]
Accordingly, every ciphertext admits a unique decomposition into a
masking component in $\operatorname{im}(B)$ and a plaintext component in
$\operatorname{im}(R)$.
The affine decomposition~\eqref{Affine_Decomposition} shows that every
plaintext determines the affine ciphertext space
\[
R\epsilon+\operatorname{im}(B).
\]
These affine spaces form a partition of $\mathbb F_p^N$ into
$p^{N-k}$ parallel affine subspaces, each containing $p^k$
ciphertexts. For every plaintext, the fresh salt selects a uniformly
distributed ciphertext from the corresponding affine space.

The direct-sum decomposition also admits a natural quotient-space
interpretation. Let
\[
q:\mathbb F_p^N
\longrightarrow
\mathbb F_p^N/\operatorname{im}(B)
\]
denote the canonical quotient map. Since
\(q(B\mu)=0\),
the affine decomposition yields
$
q(\zeta)
=
q(R\epsilon)$.
Thus the masking component vanishes upon passing to the quotient.
Accordingly, the induced map
\begin{equation}
\label{quotient_map}
\overline{R}
=
q\circ R:
\mathbb F_p^{N-k}
\longrightarrow
\mathbb F_p^N/\operatorname{im}(B)
\end{equation}
is an isomorphism. Consequently, the quotient retains precisely the
plaintext component of the ciphertext while discarding the masking
component introduced by the fresh salt.

The affine and quotient-space viewpoints are complementary. The former
describes how the fresh salt randomizes the ciphertext within an affine
space, while the latter isolates the deterministic plaintext component
by factoring out the masking subspace. Both viewpoints will be useful in
the cryptanalytic discussion that follows.

\subsection{Repeated-Input Key-Recovery Attack}
\label{Repeated_Input_Attack}

The affine structure developed in the preceding subsection naturally
suggests a chosen-plaintext attack based on repeated encryptions of a
fixed plaintext. By linearity, repeated encryptions of an arbitrary
plaintext reduce to repeated encryptions of the zero plaintext.
The resulting ciphertexts reveal the masking subspace
$\operatorname{im}(B)$, from which the attacker reconstructs the graph
matrix $A$. The graph identity (\ref{A_graph})
then yields a system of linear relations among the key coefficients.
Indeed, if
\[
\zeta^{(i)}
=
R\epsilon_0+B\mu^{(i)},
\qquad
i=1,2,\ldots,
\]
then
\[
\zeta^{(i)}-\zeta^{(1)}
=
B(\mu^{(i)}-\mu^{(1)}),
\]
so the constant plaintext contribution is eliminated. It therefore
suffices to analyze repeated encryptions of the zero plaintext.

Assume throughout that
\[\boxed{
2k\le N.}
\]
Under this hypothesis, the attacker reconstructs the secret key by the
following sequence of steps.

\paragraph{Step 1. Reconstruction of the graph matrix.}
Partition the masking matrix into the block form
\[
B=
\begin{pmatrix}
B_0\\
B_1
\end{pmatrix},
\]
where $B_0$ consists of the first $k$ rows and $B_1$ of the remaining
$N-k$ rows. Since $B_0$ is lower triangular with nonzero diagonal
entries, it is invertible. Consequently,
\[
\operatorname{im}(B)
=
\left\{
\begin{pmatrix}
u\\
Au
\end{pmatrix}
:
u\in\mathbb F_p^k
\right\},
\qquad
\text{where}\qquad
A=B_1B_0^{-1}.
\]

Suppose the attacker has obtained $k$ linearly independent ciphertexts
spanning $\operatorname{im}(B)$, and arrange them as the columns of a
matrix
\[
M=
\begin{pmatrix}
M_0\\
M_1
\end{pmatrix}.
\]
Since the columns of $M$ form another basis of
$\operatorname{im}(B)$, there exists an invertible matrix
$C\in GL_k(\mathbb F_p)$ such that
\[
M=BC.
\]
Hence
$
M_0=B_0C
$,
$
M_1=B_1C,
$
and therefore
\[
M_1M_0^{-1}
=
B_1CC^{-1}B_0^{-1}
=
B_1B_0^{-1}
=
A.
\]
Thus the graph matrix $A$ is recovered directly from any basis of the
masking subspace. The subsequent reconstruction of the secret key
therefore depends only on the recovered matrix $A$ and the arithmetic
structure of Dirichlet convolution matrices.
The identity
\begin{equation}\label{A_graph}
  B_1=AB_0
\end{equation}
encodes a system of linear relations among the key coefficients.

\paragraph{Step 2. Recovery of the initial key coefficients.}
By \eqref{quotient_map}, the induced map
$
\overline{R}$
is an isomorphism. Consequently, a chosen-plaintext attacker may
reconstruct the matrix $R$ modulo the masking subspace
$\operatorname{im}(B)$. Since the masking subspace has already been
recovered, the matrix $R$ itself is therefore known.

The columns of $R$ are precisely the last $N-k$ columns of the
Dirichlet convolution matrix $D_\kappa$. In particular, for every
\[
q\le
m=
\left\lfloor
\frac{N}{k+1}
\right\rfloor,
\]
the coefficient $\kappa_q$ appears explicitly among the entries of
$R$, since
\[
(k+1)q\le N.
\]
Thus the attacker directly recovers the initial coefficients
\[
\kappa_1,\ldots,\kappa_m.
\]
These coefficients provide the initial data required for the recursive
reconstruction of the remaining leading coefficients from the identity
(\ref{A_graph}).

\paragraph{Step 3. Recovery of the remaining coefficients.}
Throughout this section we (crucially) assume
$
N\ge 2k$.

The reconstruction of the key proceeds in two stages. We first recover the
remaining leading coefficients
\[
\kappa_{m+1},\ldots,\kappa_k,
\]
using suitably chosen higher-column equations of equation~(\ref{A_graph}).
Once all leading coefficients have been determined, the remaining tail
coefficients
\[
\kappa_{k+1},\ldots,\kappa_N
\]
follow immediately from the first-column equations of the same identity.
Thus the only nontrivial step is the recovery of the remaining leading
coefficients.

\begin{proposition}[Recovery of the Remaining Leading Coefficients]
\label{Remaining_Leading_Coefficients}

Suppose that the graph matrix $A$ is known and that the initial key
coefficients
\[
\kappa_1,\ldots,\kappa_m,
\qquad
m=\left\lfloor\frac{N}{k+1}\right\rfloor,
\]
have been recovered. Then the remaining leading coefficients
\[
\kappa_{m+1},\ldots,\kappa_k
\]
are uniquely determined.

\end{proposition}

\begin{proof}
The Graph Identity~\eqref{A_graph}
is equivalent to the componentwise relations
\[
\mathbf1_{j\mid n}\kappa_{n/j}
=
\sum_{q=1}^{\lfloor k/j\rfloor}
a_{n-k,jq}\kappa_q,
\qquad
k<n\le N,\quad
1\le j\le k.
\]
Proceed by induction for
\[
r=m+1,\ldots,k.
\]
Assume that
$
\kappa_1,\ldots,\kappa_{r-1}$
have already been recovered.
Define
\[
j_r=\left\lfloor\frac{k}{r}\right\rfloor+1,
\qquad
n_r=j_rr.
\]
Since
$
j_r>k/r$,
we have
$
n_r>k$.
Moreover,
\[
n_r
=
r\left(\left\lfloor\frac{k}{r}\right\rfloor+1\right)
\le
r\left(\frac{k}{r}+1\right)
=
k+r
\le
2k
\le
N.
\]
Therefore,
\[
k<n_r\le N,
\]
and the corresponding component equation extracted from
(\ref{A_graph})
is
\[
\kappa_r
=
\sum_{q=1}^{\lfloor k/j_r\rfloor}
a_{n_r-k,j_rq}\kappa_q.
\]
Finally, since
\[
\left\lfloor\frac{k}{j_r}\right\rfloor<r,
\]
every coefficient appearing on the right-hand side has index
strictly smaller than \(r\), and is therefore already known by the
induction hypothesis.
Thus $\kappa_r$ is uniquely determined. Proceeding successively for
$
r=m+1,m+2,\ldots,k
$
recovers all remaining leading coefficients.
\end{proof}

\begin{remark}

The recovery of the leading coefficients uses equations from several
columns of the identity (\ref{A_graph})
rather than exclusively from its first column. 
\end{remark}

\subsection{Limitations of the Repeated-Input Attack}
\label{section_invisibility}

In the preceding subsection we analyzed the repeated-input attack under
the assumption \(N\ge 2k\), and showed that the leading key coefficients
(and hence the entire key) can be reconstructed uniquely. The
complementary regime 
\[\boxed{N<2k}
\] 
exhibits a qualitatively different
phenomenon. The repeated-input attack undergoes a sharp phase transition
at the critical threshold \(N=2k\). Above this threshold every leading
coefficient is recoverable. Below it, some leading coefficients become
intrinsically unrecoverable, with the resulting ambiguity governed
entirely by the arithmetic of divisibility.

\begin{definition}
The sets of \emph{visible} and \emph{invisible} leading indices are
defined by
\[
\mathcal V_{N,k}
=
\left\{
r\in\{1,\ldots,k\}:
\exists\,j\ge2
\text{ such that }
k<jr\le N
\right\},
\]
and
\[
\mathcal I_{N,k}
=
\{1,\ldots,k\}\setminus\mathcal V_{N,k}.
\]
Equivalently,
\[
r\in\mathcal V_{N,k}
\iff
\left\lfloor\frac Nr\right\rfloor
>
\left\lfloor\frac kr\right\rfloor,
\]
whereas
\[
r\in\mathcal I_{N,k}
\iff
\left\lfloor\frac Nr\right\rfloor
=
\left\lfloor\frac kr\right\rfloor.
\]
\end{definition}

The invisible indices are precisely those whose multiples never occur in the
observation window
$
(k,N]$.
Consequently, the following results show that the arithmetic notion of
invisibility completely characterizes the blind spots of the repeated-input
attack.

\begin{lemma}[Annihilation of invisible directions]
Let
\[
A=B_1(\kappa)B_0(\kappa)^{-1},
\]
and, for \(1\le r\le k\), let
\[
E_r=B_0(e_r),
\]
where \(e_r\in\mathbb F_p^k\) is the \(r\)-th standard coordinate
vector. If \(r\) is invisible, then
\[
AE_r=0.
\]
\end{lemma}

\begin{proof}
Since \(B_0(\kappa)\) and \(E_r\) belong to the commutative algebra of
truncated Dirichlet-convolution matrices, they commute. Therefore
\[
AE_r
=
B_1(\kappa)B_0(\kappa)^{-1}E_r
=
B_1(\kappa)E_rB_0(\kappa)^{-1}.
\]
It remains to prove that
$
B_1(\kappa)E_r=0$.
For \(k<n\le N\) and \(1\le j\le k\), 
\[
\bigl(B_1(\kappa)E_r\bigr)_{n-k,j}
=
\sum_{s=1}^{k}
\bigl(B_1(\kappa)\bigr)_{n-k,s}(E_r)_{s,j}.
\]
Since
$
(E_r)_{s,j}\neq0
$
only when \(s=rj\), this becomes
\[
\bigl(B_1(\kappa)E_r\bigr)_{n-k,j}
=
\begin{cases}
\kappa_{\,n/(rj)},
   & rj\mid n \text{ and } rj\le k,\\[2mm]
0,&\text{otherwise}.
\end{cases}
\]
In particular, a nonzero entry is possible only if
$
r\mid n$.
Because \(r\) is invisible, no multiple of \(r\) lies in the observation
range
$
k<n\le N$.
Consequently,
$
B_1(\kappa)E_r=0$,
and therefore
$
AE_r=0$.
\end{proof}

Next, we have the following characterization theorem: 

\begin{theorem}[Recoverability and Ambiguity]
\label{thm:recoverability-ambiguity}
Assume that the graph matrix $A$ has been recovered from a repeated-input
attack, and normalize the key by fixing $\kappa_1$.
Then
\begin{enumerate}
\item
every visible coefficient
$
\kappa_r$,
$
r\in\mathcal V_{N,k},
$
is uniquely determined;

\item
every invisible coefficient
$
\kappa_r$,
$
r\in\mathcal I_{N,k}$,
may be chosen arbitrarily and independently without changing the graph
matrix \(A\).
\end{enumerate}
Consequently, the family of normalized keys producing the same graph
matrix has dimension
$
|\mathcal I_{N,k}|$.
\end{theorem}

\begin{proof}
We work in the space of leading-coordinate perturbations
\(\mathbb F_p^k\). Let
\[
\delta=(\delta_1,\ldots,\delta_k)
=
(\widetilde\kappa_1-\kappa_1,\ldots,
\widetilde\kappa_k-\kappa_k),
\]
where \(\widetilde\kappa\) is another key producing the same graph
matrix \(A\). Let \(B_0(\delta)\) denote the corresponding truncated
Dirichlet-convolution matrix, and define the linear compatibility
operator
\[
\mathcal L_A(\delta)
=
B_1(\delta)-AB_0(\delta).
\]
Since
\(B_1(\widetilde\kappa)=AB_0(\widetilde\kappa)\)
and
\(B_1(\kappa)=AB_0(\kappa)\),
linearity gives
\[
\mathcal L_A(\delta)
=
B_1(\delta)-AB_0(\delta)
=
B_1(\widetilde\kappa)-B_1(\kappa)-AB_0(\delta)
=
0.
\]
Thus two normalized keys produce the same graph matrix if and only if
their difference belongs to \(\ker\mathcal L_A\).
We first identify a family of kernel vectors. Let
\(r\in\mathcal I_{N,k}\). Then
\(B_1(e_r)=0\), since no multiple of \(r\) lies in the observation range
\((k,N]\). By the preceding lemma,
\(AE_r=0\), where \(E_r=B_0(e_r)\). Hence
\[
\mathcal L_A(e_r)
=
B_1(e_r)-AE_r
=
0,
\]
so that
\[
\operatorname{span}
\{e_r:r\in\mathcal I_{N,k}\}
\subseteq
\ker\mathcal L_A.
\tag{1}
\label{eq:first-kernel-inclusion}
\]

It remains to prove the reverse inclusion. Let
\(\delta\in\ker\mathcal L_A\), with \(\delta_1=0\). We prove, by
induction on \(r\), that
\[
\delta_r=0
\qquad
\text{for every }r\in\mathcal V_{N,k}.
\]

Let \(r\in\mathcal V_{N,k}\), and choose \(j\ge2\) such that
\(k<jr\le N\). Since \(\mathcal L_A(\delta)=0\), we have
\(B_1(\delta)=AB_0(\delta)\). Taking the \((jr-k,j)\)-entry gives
\[
\delta_r
=
\bigl(AB_0(\delta)\bigr)_{jr-k,j}.
\]
Writing
\[
B_0(\delta)
=
\sum_{q=1}^k \delta_qE_q,
\]
we obtain
\[
\delta_r
=
\sum_{q=1}^k
\delta_q\,(AE_q)_{jr-k,j}.
\]

We now show that
\[
(AE_q)_{jr-k,j}=0
\qquad
\text{whenever }q\ge r.
\]
Since
\[
A
=
B_1(\kappa)B_0(\kappa)^{-1},
\]
and the Dirichlet-convolution matrices commute, we have
\[
AE_q
=
B_1(\kappa)E_qB_0(\kappa)^{-1}.
\]
Consequently,
\[
(AE_q)_{jr-k,j}
=
\sum_{s=1}^{k}
(B_1(\kappa)E_q)_{jr-k,s}
(B_0(\kappa)^{-1})_{s,j}.
\]
For a summand on the right-hand side to be nonzero, one must have
\(sq\le k\). Moreover,
\[
(B_0(\kappa)^{-1})_{s,j}\neq0
\]
implies \(j\mid s\). Writing \(s=jd\), we obtain
\[
jdq\le k<jr.
\]
Hence \(dq<r\), and therefore \(q<r\). This proves the asserted
vanishing for \(q\ge r\).

If \(q\in\mathcal I_{N,k}\), then \(AE_q=0\) by the preceding lemma.
It follows that
\[
\delta_r
=
\sum_{\substack{q<r\\ q\in\mathcal V_{N,k}}}
\delta_q\,(AE_q)_{jr-k,j}.
\tag{2}
\label{eq:visible-induction}
\]
By the induction hypothesis, \(\delta_q=0\) for every visible
\(q<r\). Therefore \(\delta_r=0\). Proceeding upward through the
visible indices yields
\[
\delta_r=0
\qquad
\text{for every }r\in\mathcal V_{N,k}.
\]
Together with
\eqref{eq:first-kernel-inclusion},
this proves that
\[
\ker\mathcal L_A
=
\operatorname{span}
\{e_r:r\in\mathcal I_{N,k}\}.
\]
Therefore every normalized leading vector compatible with \(A\) is of
the form
\[
\widetilde\kappa^{(0)}
=
\kappa^{(0)}
+
\sum_{r\in\mathcal I_{N,k}}
c_r e_r,
\qquad
c_r\in\mathbb F_p,
\]
where
\[
\kappa^{(0)}
=
(\kappa_1,\ldots,\kappa_k).
\]
It follows that the visible leading coefficients are uniquely
determined, while the invisible leading coefficients may be chosen
arbitrarily and independently. Consequently, the family of normalized
keys producing the same graph matrix has dimension
$
|\mathcal I_{N,k}|.
$
\end{proof}

\begin{remark}
By the preceding theorem, the repeated-input ambiguity is confined to
the invisible leading coefficients. These coefficients constitute the
free parameters of the reconstruction problem: every assignment of
their values determines a unique normalized key compatible with the
given graph matrix. The visible leading coefficients and the remaining
tail coefficients are then uniquely determined by the reconstruction
equations. Consequently, the family of normalized keys compatible with
a given graph matrix is an affine space of dimension
\( |\mathcal I_{N,k}| \).
\end{remark}

The following example provides a simple illustration of the phase
transition described above. It demonstrates that, once $2k>N$, the
repeated-input ambiguity may generate a space of invisible coefficients
of significant dimension.

\begin{example}
Consider the parameters
\[
N=550,
\qquad
k=450.
\]
The observation window is
\[
(k,N]=(450,550].
\]
A leading index $r$ is visible precisely when one of its multiples lies in
this interval. Equivalently,
\[
\left\lfloor\frac Nr\right\rfloor
>
\left\lfloor\frac kr\right\rfloor.
\]
Consequently, the invisible indices exhibit a band structure reflecting
the arithmetic nature of invisibility: each band corresponds to a
distinct value of $j$ for which
\[
jr\le k<N<(j+1)r.
\]
The invisible indices are therefore
\[
111\le r\le112 \qquad (j=4),
\]
\[
138\le r\le150 \qquad (j=3),
\]
\[
184\le r\le225 \qquad (j=2),
\]
and
\[
276\le r\le450 \qquad (j=1).
\]
Hence there are
$
232
$
invisible coefficients and
$
217
$
visible coefficients.

The invisible coefficients and the associated ambiguity space are easily
explored numerically, reflecting the constructive nature of
Theorem~\ref{thm:recoverability-ambiguity}. Rather than merely asserting
the existence of ambiguity, the theorem identifies its precise location
and provides an explicit basis for the corresponding nullspace.

\end{example}

\section{Efficient Algorithms for Dirichlet Convolution}
\label{Section_fast_algorithm}

In this section we show that Dirichlet convolution admits efficient
algorithms for both multiplication and inversion. These algorithms
exploit the sparse structure of the Dirichlet matrices introduced in
Section~\ref{Section_DCM}. For a sequence
$a=[a_1,\ldots,a_N]$, let $D_a$ denote the corresponding Dirichlet
matrix defined by \eqref{D_matrix}. Matrix multiplication realizes the
Dirichlet algebra:
\[
D_aD_b=D_{a*b}.
\]

The number of nonzero entries in the $n$th row of $D_a$ is exactly the
number of divisors of $n$, denoted $d(n)$. Consequently, the total
number of nonzero entries is
\[
\sum_{n\le N} d(n).
\]
Dirichlet's asymptotic formula (see \cite[Theorem~3.3]{Apostol2})
states that
\begin{equation}\label{Dir_estimate}
\sum_{n\le N} d(n)
=
N\log N+(2\gamma-1)N+O(\sqrt N),
\end{equation}
where $\gamma=0.57721566\ldots$ is the Euler--Mascheroni constant.
Hence every Dirichlet matrix contains only $O(N\log N)$ nonzero
entries. Consequently, Dirichlet convolution can be evaluated using
$O(N\log N)$ arithmetic operations.

Previous work has shown how to eliminate operations involving entries
that are necessarily zero. The resulting algorithms apply directly
over finite fields and are included here for completeness.

\medskip

The convolution $c = a*b$ can be computed efficiently by the following algorithm:
\paragraph{Fast multiplication.}\label{Dirichlet_fast_mult}
Input vectors $a$ and $b$. Initiate  vector $c = [0,0,\ldots 0]$ ($N$ zeros). Next,
\[
\begin{array}{lll}
\mbox{for } k = 1:N \\
\quad\quad    c(k:k:N) = c(k:k:N) + b(k)\cdot a(1:\lfloor N/k \rfloor );\\
\mbox{end}
\end{array}
\]
(Note: operations $+$ and $\cdot$ are performed in the underlying finite field.)
When the for-loop is completed $c= a*b$.

\medskip

Here,  $(k:k:N)$ denotes the set of indices $(k, 2k,\ldots ,  \lfloor N/k\rfloor \, k)$. The algorithm requires $O(N\log N)$ arithmetic operations.
 Note that because of commutativity the roles of $a$  and $b$ can be reversed.

\medskip

The inverse $b = a^{*-1}$ can be computed recursively by exploiting the triangular structure of the matrix representation.

\paragraph{Fast inversion.}\label{alg_D_inverse}
Input vector $a$ with $a_1\neq 0$. Initiate $b = [a_1^{-1},0,0,\ldots 0]$ ($N-1$ zeros). Next,
\begin{equation}\label{Dirichlet_inverse}
\begin{array}{lll}
\mbox{for } k = 1:\lfloor N/2\rfloor  \\
\quad\quad    ind = k:k:N;
\quad    ind(1) = []; \\
\quad\quad (\mbox{the first index in the series is dropped}) \\
\mbox{end}
\end{array}
\end{equation}
(Again, operations $+$ and $\cdot$ are performed in the underlying finite field.) 
When the for-loop is completed, $b$ is the Dirichlet inverse of $a$, i.e. $a*b = e$.

\medskip 

The computational complexity of the inversion algorithm is likewise $O(N\log N)$.

\begin{remark}
Both algorithms require only the storage of the input and output vectors; no auxiliary matrices need be formed. In particular, the preallocated space for $b$ remains fixed throughout the computation. Such algorithms are sometimes referred to as lifting schemes. This memory efficiency becomes especially advantageous for large problem sizes.
\end{remark}

\section{Practical Considerations}
\label{sec:practical_considerations}

Advances in computational methods, including quantum algorithms
\cite{NielsenChuang}, motivate the search for cryptographic mechanisms
whose security relies, whenever possible, on information-theoretic
rather than computational principles. The security mechanism developed
here follows this philosophy by eliminating exploitable statistical
dependence between the plaintext and the ciphertext. Once such
dependence has been shown to vanish, advances in classical, quantum, or
AI-assisted algorithms cannot exploit statistical information that is
not present.

The generalized Dirichlet convolution developed in this paper should be
viewed as a cryptographic primitive rather than a complete communication
protocol. Its algebraic structure---including the underlying group
operation, efficient inversion, and composability of independent
encryption stages---naturally supports a variety of higher-level
protocols and deployment architectures. While different applications
place different requirements on these surrounding protocols, the
underlying primitive remains unchanged.

The examples below illustrate the practicality and flexibility of this
framework rather than prescribe specific implementations. They
demonstrate how the algebraic properties established in the preceding
sections translate naturally into deployment strategies with attractive
security and computational characteristics.

\subsection{Storage versus Communication}
\label{subsec:storage_communication}

The generalized Dirichlet-convolution primitive is applicable to both
persistent storage and secure communication, although the surrounding
protocols differ substantially. Storage applications emphasize
long-term confidentiality and efficient processing of large data
objects, whereas communication protocols additionally require
synchronization, authentication, state evolution, and robustness
against packet loss or replay attacks.

For storage, a data object may be encrypted together with fresh salt
and retained in encrypted form for an extended period. Large or
multidimensional data sets may be partitioned into frames, rows,
columns, or other natural units and processed independently or through
multiple encryption stages.

Communication protocols maintain a shared cryptographic state between
the communicating parties. This state may evolve deterministically or
incorporate fresh randomness generated during successive exchanges.
Thus, while storage and communication employ the same underlying
Dirichlet-convolution primitive, they naturally give rise to distinct
protocol architectures.

\subsection{Primitive versus Protocol}
\label{subsec:primitive_protocol}

A distinctive feature of the construction is that independent
encryption stages compose naturally. If
\[
\zeta
   =\kappa^{(2)}
      *\bigl(\kappa^{(1)}*\alpha\bigr),
\]
then associativity and commutativity imply
\[
\zeta
   =\bigl(\kappa^{(2)}*\kappa^{(1)}\bigr)
      *\alpha.
\]
Thus independently generated key elements may be combined without
changing the mathematical form of the encryption operation. This
property naturally supports layered encryption, distributed key
management, and evolving encryption states.

A distinctive feature of the present construction is that composition
requires no separate cryptanalytic treatment: by associativity,
successive Dirichlet encryption stages reduce to a single encryption
with the composite key. Thus the resulting transformation remains
within the same algebraic family and is subject to the same analysis
developed above. This should be distinguished from composition at the
protocol level, where security of the constituent protocols does not
in general imply security of their composition and requires additional
considerations; see, for example, \cite{BroadbentKarvonen}.

\subsection{Deployment Architectures}
\label{subsec:deployment_architectures}

The nonlocal character of Dirichlet convolution naturally supports a
variety of deployment architectures. For multidimensional data, one
may apply Dirichlet encryption successively to rows and columns, using
fresh salt variables and, where appropriate, independent keys at each
stage. The row transform mixes information within individual rows,
while the subsequent column transform propagates these dependencies
across the second dimension.

This two-stage construction substantially changes the geometry of a
potential attack. Recovering the action of either transform generally
requires accounting for the other. The same principle extends
immediately to higher-dimensional arrays, collections of frames, and
other naturally partitioned data structures.

A second architecture interleaves Dirichlet convolution with a secret
permutation:
\[
\bm{\alpha}
 \xrightarrow{\;D_{\bm{\kappa}^{(1)}}\;}
\bm{\eta}
 \xrightarrow{\;\text{Secret permutation}\;}
\bm{\xi}
 \xrightarrow{\;D_{\bm{\kappa}^{(2)}}\;}
\bm{\zeta}.
\]
The first convolution suppresses readily observable statistical
features of the input, the permutation conceals the divisor geometry,
and the second convolution performs a further nonlocal transformation.
Since an arbitrary permutation cannot, in general, be absorbed into a
change of the convolution key, the two operations provide genuinely
complementary mechanisms of diffusion.

More generally, recovering hidden transformations interleaved with
structured cryptographic maps may lead to nontrivial equivalence
problems; related questions arise, for example, in the study of
extended-affine equivalence of vectorial Boolean functions
\cite{CanteautCouvreurPerrin}.

The number and order of stages, the choice of keys and salts, and the
partition of the data may all be adapted to the security and
performance requirements of a particular application.

\subsection{Dynamic Key Evolution}
\label{subsec:dynamic_keys}

The group structure naturally supports dynamic evolution of the
encryption state. If the key at stage \(j\) is \(\kappa_j\), one may
define
\[
\kappa_{j+1}
   =\rho_j*\kappa_j,
\]
where \(\rho_j\) is a fresh update element generated or exchanged
according to the surrounding protocol. Encryption at the next stage
then proceeds using \(\kappa_{j+1}\).

Associativity and commutativity permit successive updates to be
composed efficiently and, in distributed settings, contributed by
multiple parties. This naturally suggests communication protocols
based on synchronized evolution of the encryption state, including
constructions resembling a cryptographic ratchet.

The design and analysis of such protocols require additional
considerations, including authentication, synchronization, recovery
from state compromise, and fault tolerance, and therefore lie beyond
the scope of the present paper. Nevertheless, the algebraic
composability of the primitive makes dynamic key evolution a natural
extension of the present framework.

\subsection{Prototype Implementation and Performance}
\label{subsec:prototype_performance}

Prototype implementations have been developed in Python, Rust, and
CUDA to assess the practical feasibility of generalized Dirichlet
encryption. Representative performance measurements are summarized in
Table~\ref{tab:performance}. Among the current implementations, Rust
consistently achieved the highest throughput. These measurements
demonstrate practical feasibility rather than platform-independent
performance benchmarks.

\begin{table}[t]
\centering
\caption{Representative encryption times for current prototype implementations of generalized Dirichlet encryption. The reported values are intended to illustrate practical feasibility rather than to provide platform-independent performance benchmarks.}
\label{tab:performance}

\begin{tabular}{lccc}
\hline
\textbf{Test case} & \textbf{Python} & \textbf{Rust} & \textbf{CUDA} \\
\hline
512$\times$512 grayscale image  & 80.03 ms & \textbf{23.00 ms} & 56.90 ms \\
1024$\times$1024 grayscale image & 153.44 ms & \textbf{107.00 ms} & 140.25 ms \\
\hline
673$\times$1102 DICOM$^\dagger$ & 22.73 ms & \textbf{3.45 ms}$^\ast$ & 4.23 ms \\
\hline
\end{tabular}

\vspace{0.3em}

\footnotesize
$^\dagger$ 28-frame, 3-channel medical image.\\
$^\ast$ 12-core multithreaded Rust implementation.
\end{table}

Figure~\ref{fig:medical_example} illustrates generalized Dirichlet
encryption applied to representative medical-image data. Encryption
removes visually recognizable anatomical structure while
substantially flattening the empirical intensity histogram. Exact
recovery is obtained upon decryption. The histograms provide an
intuitive illustration of the masking properties established
theoretically in the preceding sections.

\begin{figure}[t]
\centering

\begin{minipage}[t]{0.47\textwidth}
\centering
\includegraphics[width=\linewidth]{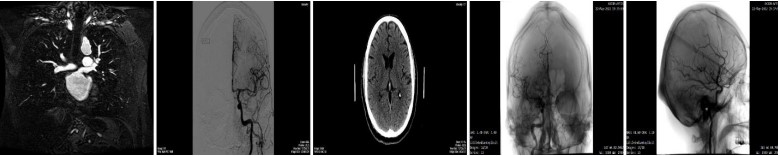}

{\small (a) Original images}
\end{minipage}
\hfill
\begin{minipage}[t]{0.47\textwidth}
\centering
\includegraphics[width=\linewidth]{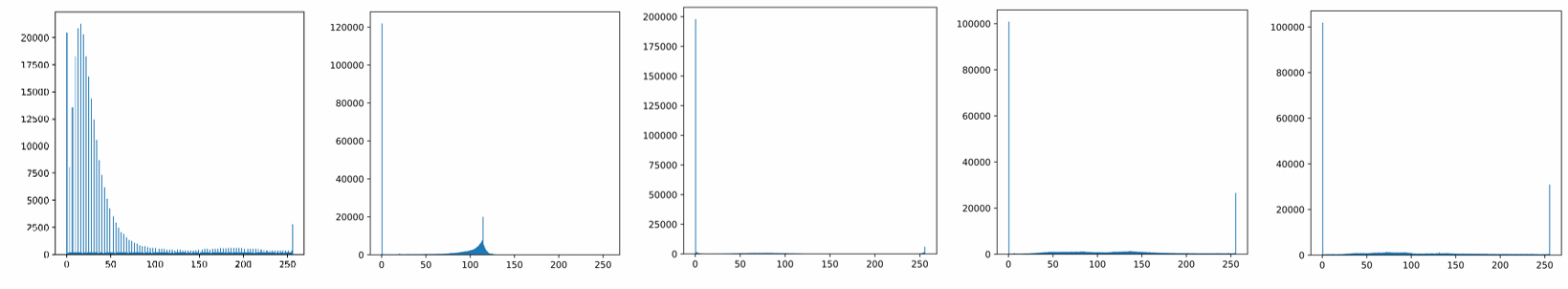}

{\small (b) Original histograms}
\end{minipage}

\vspace{0.5em}

\begin{minipage}[t]{0.47\textwidth}
\centering
\includegraphics[width=\linewidth]{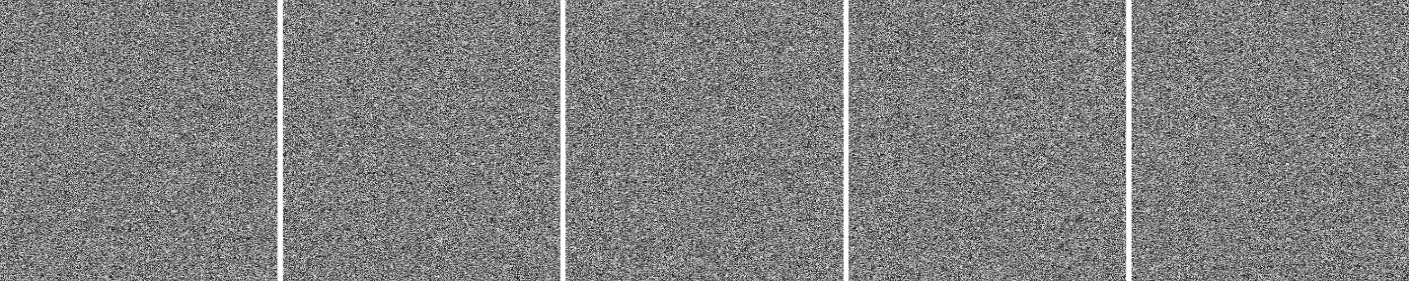}

{\small (c) Encrypted images}
\end{minipage}
\hfill
\begin{minipage}[t]{0.47\textwidth}
\centering
\includegraphics[width=\linewidth]{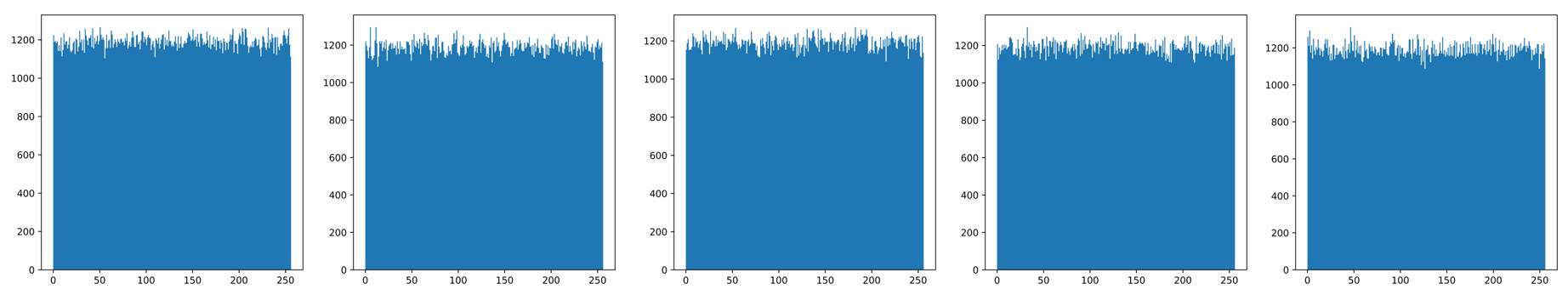}

{\small (d) Histograms of encrypted images}
\end{minipage}

\caption{
Representative prototype implementation of generalized Dirichlet
encryption applied to medical-image data. The encrypted images
exhibit no recognizable features, while their corresponding
ciphertext histograms are flat. Exact recovery of the original images
is obtained upon decryption. The histograms provide an intuitive
illustration of the masking properties established theoretically in
the preceding sections.
}
\label{fig:medical_example}
\end{figure}

These prototype implementations demonstrate that the proposed
algebraic framework admits efficient software realization while
preserving exact invertibility. Further optimization, hardware
acceleration, and systematic benchmarking remain natural directions
for future work.

\ifanonymous
\else
\section*{Acknowledgment}
The authors gratefully acknowledge the support of The Royal University Hospital Foundation.
Francis Bui gratefully acknowledges the support of the Natural Sciences and Engineering Research Council of Canada (NSERC).
\fi

\appendix

\section{The Exceptional Binary Case}
\label{App_binary_case}

The perfect-masking argument developed in the main text relies on
multiplication by uniformly distributed nonzero elements of
\(\mathbb F_p\). The binary case is exceptional: products of independent
binary variables are biased toward zero, and the resulting mask need
not be uniformly distributed. We record here the corresponding
coordinate-wise mutual information.

Let
\[
\zeta
=
\kappa*\epsilon
\pmod 2,
\]
where \(\kappa_1=1\), and suppose that the remaining key and plaintext
coordinates are independent and uniformly distributed over
\(\mathbb F_2\). For \(n>1\),
\[
\zeta_n
=
\epsilon_n+\gamma_n
\pmod 2,
\qquad
\gamma_n
=
\sum_{\substack{d\mid n\\ d>1}}
\kappa_d\epsilon_{n/d}
\pmod 2.
\]

Let \(d(n)\) denote the number of positive divisors of \(n\). Since
\(\gamma_n\) is the parity of \(d(n)-1\) independent Bernoulli products,
\[
P(\gamma_n=0)
=
\frac12\left(1+2^{-d(n)+1}\right),
\qquad
P(\gamma_n=1)
=
\frac12\left(1-2^{-d(n)+1}\right).
\]
Consequently,
\begin{equation}
\label{Izeta_binary}
I(\zeta_n;\epsilon_n)
=
\frac{1-2^{-d(n)+1}}{2}
\log_2\!\left(1-2^{-d(n)+1}\right)
+
\frac{1+2^{-d(n)+1}}{2}
\log_2\!\left(1+2^{-d(n)+1}\right).
\end{equation}

The leakage is largest when \(n\) is prime. In that case,
\(d(n)=2\) and
\[
I(\zeta_n;\epsilon_n)\approx 0.1887
\]
bits. For \(n=q^2\), with \(q\) prime, the value decreases to
approximately \(0.0456\) bits, while for \(n=qr\), with distinct
primes \(q\) and \(r\), it is approximately \(0.0113\) bits. More
generally, the mutual information decays rapidly as \(d(n)\) increases.

Thus the binary construction exhibits small but nonzero
coordinate-wise leakage. Any use of a fixed binary key therefore
requires a separate analysis of cumulative leakage under key reuse.
This contrasts with the \(p>2\) setting of the main text, where the
stated hypotheses yield exact coordinate-wise masking. The exceptional
behavior of the binary field thus provides additional motivation for
developing the theory over general finite fields \(\mathbb F_p\),
\(p>2\), where perfect masking becomes possible.

\end{document}